\PassOptionsToPackage{dvipsnames}{xcolor}
\documentclass[runningheads,envcountsame]{llncs}
\usepackage[title]{appendix}
\usepackage{graphicx}
\usepackage{latexsym} %for sqsubset
\usepackage{verbatim}
\usepackage{cite}
\usepackage{amsmath}
\usepackage{amsfonts}
\usepackage{amssymb}
\usepackage{tikz}
\usetikzlibrary{positioning, shapes, arrows, snakes}
\usetikzlibrary{backgrounds}
\usepackage{tikz-cd}
\usepackage{todonotes}
\usepackage[dvipsnames]{xcolor}
\usepackage{mathpartir}
\usepackage{bussproofs}
\usepackage{enumitem}

\usepackage{xcolor}
\newcommand{\ld}[1]{\textcolor{MidnightBlue}{#1}}
\newcommand{\rid}[1]{\textcolor{RedOrange}{#1}}
\newcommand{\rd}[1]{\textcolor{RedOrange}{#1'}}
\newcommand{\mt}[2]{
\ifthenelse{\equal{#2}{}}{\langle #1,\rd{#1}\rangle}{\langle #1_{#2},\rd{#1_{#2}}\rangle}
}
\newcommand{\mtsq}[2]{
\ifthenelse{\equal{#2}{}}{[#1,\rd{#1}]}{[#1_{#2},\rd{#1_{#2}}]}
}

\begin{document}
\title{Only Connect, Securely\thanks{Supported by Indo-Japanese project \textit{Security in the IoT Space}, DST, Govt of India.}}
%
%\titlerunning{Abbreviated paper title}
% If the paper title is too long for the running head, you can set
% an abbreviated paper title here
%
\author{Chandrika Bhardwaj \and
Sanjiva Prasad} %\and Third Author\inst{2,3}\orcidID{2222--3333-4444-5555}}
%
%\authorrunning{F. Author et al.}
% First names are abbreviated in the running head.
% If there are more than two authors, 'et al.' is used.
%
\institute{Indian Institute of Technology Delhi, India \\%\and
%Springer Heidelberg, Tiergartenstr. 17, 69121 Heidelberg, Germany
%\email{lncs@springer.com}\\
% \url{http://www.springer.com/gp/computer-science/lncs} %\and
%ABC Institute, Rupert-Karls-University Heidelberg, Heidelberg, Germany\\
\email{\{chandrika,sanjiva\}@cse.iitd.ac.in}}
\maketitle              % typeset the header of the contribution
\begin{abstract}
The \textit{lattice model} proposed by Denning in her seminal work provided secure information flow analyses with an intuitive and uniform mathematical foundation.
Different organisations, however, may employ quite different security lattices.
In this paper, we propose a connection framework that permits different organisations to exchange information while maintaining both security of information flows as well as their autonomy in formulating and maintaining security policy.
Our prescriptive framework is based on the rigorous mathematical framework of \textit{Lagois connections} given by Melton, together with a simple operational model for transferring object data between domains.
The merit of this formulation is that it is simple, minimal, adaptable and intuitive, and provides a formal framework for establishing secure information flow across autonomous interacting organisations.
We show that our framework is semantically sound, by proving that the connections proposed preserve standard correctness notions such as non-interference.

\keywords{Security class lattice \and Information flow
\and
Lagois connection \and Atomic operations \and
Non-interference.
}
\end{abstract}

\section{Introduction}
%Background
Denning's seminal work \cite{Denning76} proposed \textit{complete lattices} as the appropriate mathematical framework for questions regarding \textit{secure information flow} (SIF).
An information flow model (IFM) is characterised as
$\langle N, P, SC, \sqcup, \sqsubseteq \rangle$ where:
\textit{Storage objects} in $N$ are assigned \textit{security levels} drawn from a (finite) complete lattice $SC$. $P$ is a set of processes (also assigned security classes as clearances).
The partial ordering $\sqsubseteq$ represents \textit{permitted flows} between classes; reflexivity and transitivity capture intuitive aspects of information flow; antisymmetry helps avoid redundancies in the framework, and the join operation $\sqcup$ succinctly captures the combination of information belonging to different security classes in arithmetic, logical and computational operations.
This lattice model provides an abstract uniform framework that identifies the commonalities of the variety of analyses for different applications -- \textit{e.g.}, confidentiality and trust -- whether at the \textit{language} level or a \textit{system} level.
In the ensuing decades, the vast body of secure information flow analyses has been built on these mathematical foundations, with the development of a plethora of static and dynamic analysis techniques for programming languages \cite{sabelfeld2003language,myers1999jflow,Pottier2003-FlowCaml,liu2017fabric,roy2009laminar,Lourenco2015-ug}, operating systems \cite{Krohn2007-aa,zeldovich2006-osdi,cheng2012aeolus,efstathopoulos2005asbestos,roy2009laminar}, databases \cite{schultz2013ifdb}, and hardware architectures \cite{ferraiuolo2018hyperflow, zhang2015secVerilog-asplos}, etc.
The soundness of this lattice model was expressed in terms of semantic notions of system behaviour, for instance, as properties like non-interference \cite{DBLP:conf/sp/GoguenM82a} by Volpano \textit{et al} \cite{DBLP:journals/jcs/VolpanoIS96} and others.
Alternative semantic notions of security such as safety properties have been proposed as well, \textit{e.g.},  \cite{Boudol2008SIFasSafety}, but for brevity we will not explore these further.

% Objective
The objective of this paper is to propose a simple way in which large-scale distributed secure systems can be built by connecting components in a secure and modular manner.
Our work begins with the observation that large information systems are not monolithic:  Different organisations define their own information flow policies independently, and subsequently collaborate or federate with one another to exchange information.
In general, the security classes and the lattices of any two organisations may be quite different --- \textit{there is no single universal security class lattice}.
Moreover, \textit{modularity} and \textit{autonomy} are important requirements since each organisation would naturally wish to retain control over its own security policies and the ability to redefine them.
Therefore, fusing different lattices by taking their union is an unsatisfactory approach, more so since the security properties of application programs would have to be re-established in this possibly enormous lattice.

When sharing information, most organisations limit the cross-domain communications to a limited set of security classes (which we call \textit{transfer} classes).
In order to ensure that shared data are not improperly divulged, two organisations usually negotiate agreements or memorandums of understanding (MoUs), promising that they  will respect the security policies of the other organisation.
We argue that a good notion of secure connection  should require reasoning only about those flows from only the transfer classes mentioned in a MoU.
Usually, cross-domain communication involves downgrading the security class of privileged information to public information using primitives such as encryption, and then upgrading the information to a suitable security class in the other domain.
Such approaches, however, do not gel well with correctness notions such as non-interference.
Indeed the question of how to translate information between security classes of different lattices is interesting \cite{deng2017secchisel}.

\paragraph{Contributions of this paper.} \ \
In this paper, we propose a simple framework and sufficient conditions under which secure flow guarantees can be enforced without exposing the complexities and details of the component information flow models.
The framework consists of (1) a way to connect security classes of one organisation to those in another while satisfying intuitive requirements;  (2) a simple language that extends  the operations within an organisation with primitives for transferring data between organisations; and (3) a type system and operational model for these constructs, which we use to establish that the framework conserves security.

In \S\ref{sec:Bi-DirFlow}, we first identify, using intuitive examples,
violations in secure flow that may arise when two secure systems are permitted to exchange information in both directions.
Based on these lacunae, we formulate \textit{security} and \textit{precision} requirements for secure bidirectional flow.
We then propose a framework that guarantees the absence of such policy violations, without impinging on the autonomy of the individual systems,  without the need for re-verifying the security of the application procedures in either of the domains, and confining the analysis to only the transfer classes involved in potential exchange of data.
Our approach is based on \textit{monotone functions} and an elegant theory of \textit{connections} \cite{MELTON1994lagoisconnections} between the security lattices.
Theorem \ref{thm:secureconnection} shows that\textit{ Lagois connections} between the security lattices satisfy the security and precision requirements.

We present, in \S\ref{sec:model}, a minimal operational language consisting of a small set of \textit{atomic primitives} to effect the transfer of data between domains.
The framework is simple and can be adapted for establishing secure connections between distributed systems at any level of abstraction (language, system, database, ...).
We assume each domain uses \textit{atomic transactional operations} for object manipulation and intra-domain computation.
The primitives of our model include reliable communication between two systems, transferring object data in designated \textit{output} variables of one domain to designated \textit{input} variables of a specified security class in the other domain.
We also assume a generic set of operations in each domain for copying data from input variables to domain objects, and from domain objects to output variables.
To avoid interference between inter-domain communication and the computations within the domains, we assume that the sets of designated input and output variables are all mutually exclusive of one another, and also with the program/system variables used in the computations within each domain.
Thus by design we avoid the usual suspects that cause interference and insecure transfer of data.
The operational description of the language consists of the primitives
together with their execution rules (\S\ref{sec:operations}).

The correctness of our framework is demonstrated by expressing soundness (with respect to the operational semantics) of a type system (\S\ref{sec:typing}), stated in terms of the security lattices and their connecting functions.
In particular, Theorem \ref{thm:soundness} shows the standard semantic property of \textit{non-interference} in both domains holds of all operational behaviours.
We adapt and \textit{extend} the approach taken by Volpano \textit{et al} \cite{DBLP:journals/jcs/VolpanoIS96} to encompass systems coupled using the Lagois connection conditions (and assuming atomicity of the data transfer operations), to show that security is conserved.
We believe that our formulation is general enough to be applicable to other behavioural notions of secure information flow as a safety property \cite{Boudol2008SIFasSafety}.

In \S\ref{sec:related}, we briefly review some related work.
We conclude in \S\ref{sec:conclusion} with a discussion on our approach and directions for future work.

\section{Lagois Connections and All That}\label{sec:Bi-DirFlow}

\paragraph{Motivating Examples.} \ \
Consider a university system in which students study in semi-autonomously administered colleges (one such is $\ld{C}$) that are affiliated to a university ($\rid{U}$). The  university also has ``university professors'' with whom students can take classes.
We assume each institution has established the security of its information flow mechanisms and policies.

We first observe that formulating an agreement \textit{between} the institutions that respects the flow policies within the institutions is not entirely trivial.
Consider an arrangement where the college \textit{Faculty} and \textit{University Faculty} can share information (say, course material and examinations), and the \textit{Dean of Colleges} in the University can exchange information (e.g., students' official grade-sheets) with the college's \textit{Dean of Students}.
Even such an apparently reasonable arrangement suffers from insecurities, as illustrated in Fig. \ref{fig:Bi-breach-equiv-classes} by the flow depicted using dashed red arrows, where information can flow from the college's \textit{Faculty} to the college's \textit{Dean of Students}. (Moral: internal structure of the lattices matters.)
% \vspace{-0.9cm}
\begin{figure}[t]%[!b]
% \centering
    \begin{minipage}[t]{.45\textwidth}
        % \centering
    \begin{tikzpicture}[framed,->,node distance=0.9cm,on grid]
    \title{W and D}
    \node(T2)   {\scriptsize$\top 2$};
    \node(T1)  [xshift=-3cm]  {\scriptsize$\top 1$};
    \node(Dir1) [below of = T1] {\scriptsize$CollegePrincipal$};
    \node(D1) [below left of = Dir1] {\scriptsize$Dean\ (F)$};
    \node(F1) [below of = D1] {\scriptsize$Faculty$};
    \node(DS1) [xshift=-2.3cm,yshift=-1.7cm] {\scriptsize$Dean\ (S)$};
    \node(S1) [below right of = F1] {\scriptsize$Student$};
    \node(B1)  [below of=S1]  {\scriptsize$\bot 1$};
    \node(Sec2)  [below of = T2] {\scriptsize$Chancellor$};
    \node(AS2)  [below of=Sec2] {\scriptsize$Vice\ Chancellor$};
    \node(Dir2)  [below of=AS2] {\scriptsize$Dean(Colleges)$};
    \node(E2)  [below of=Dir2] {\scriptsize$Univ.Fac.$};
    \node(B2)  [below of=E2]  {\scriptsize$\bot 2$};
    \draw [Fuchsia, ->] (Dir1) to (T1);
    \draw [Fuchsia, ->] (D1) to (Dir1);
    \draw [Fuchsia, ->] (F1) to (D1);
    \draw [Fuchsia, ->] (S1) to (F1);
    \draw [Fuchsia, ->] (S1) to (DS1);
    \draw [Fuchsia, ->] (DS1) to (Dir1);
    \draw [Fuchsia, ->] (B1) to (S1);
    \draw [Fuchsia, ->] (B2) to (E2);
    \draw [Fuchsia, ->] (E2) to (Dir2);
    \draw [Fuchsia, ->] (Dir2) to (AS2);
    \draw [Fuchsia, ->] (AS2) to (Sec2);
    \draw [Fuchsia, ->] (Sec2) to (T2);
    \draw [OliveGreen, thick, <->](F1) [bend right = 5] to (E2);
    \draw [OliveGreen, thick, <->](DS1) to (Dir2);
    \draw [OliveGreen, thick, ->](Dir1) [bend left = 10] to (AS2);
    \draw [red,  thick, dashed] (F1) [bend left = 7] to (E2);
    \draw [red,  thick, dashed] (E2) [bend left = 35] to (Dir2);
    \draw [red,  thick, dashed] (Dir2) [bend left = 15] to (DS1);
    \draw (-3,-2.2) ellipse (1.4cm and 2.4cm);
    \draw (0,-2.2)[red] ellipse (1.2cm and 2.5cm);
    \end{tikzpicture}
    \caption{\small Green arrows represent permitted flows according to the information exchange arrangement between a college and a university. Red dashed arrows highlight a \textit{new} flow that is a security violation. \label{fig:Bi-breach-equiv-classes}}
    \end{minipage}
    \quad \quad
    \begin{minipage}[t]{.45\textwidth}
        % \centering
    \begin{tikzpicture}[framed,->,node distance=0.9cm,on grid]
    \title{W and D}
    \node(T2)   {\scriptsize$\top 2$};
    \node(T1)  [xshift=-3cm]  {\scriptsize$\top 1$};
    \node(Dir1) [below of = T1] {\scriptsize$CollegePrincipal$};
    \node(D1) [below left of = Dir1] {\scriptsize$Dean\ (F)$};
    \node(F1) [below of = D1] {\scriptsize$Faculty$};
    \node(DS1) [xshift=-2.3cm,yshift=-1.7cm] {\scriptsize$Dean\ (S)$};
    \node(S1) [below right of = F1] {\scriptsize$Student$};
    \node(B1)  [below of=S1]  {\scriptsize$\bot 1$};
    \node(Sec2)  [below of = T2] {\scriptsize$Chancellor$};
    \node(AS2)  [below of=Sec2] {\scriptsize$Vice\ Chancellor$};
    \node(Dir2)  [below of=AS2] {\scriptsize$Dean(Colleges)$};
    \node(E2)  [below of=Dir2] {\scriptsize$Univ.Fac.$};
    \node(B2)  [below of=E2]  {\scriptsize$\bot 2$};
    \draw [blue,  thick] (S1) to (E2);
    \draw [green,  thick] (F1) to (E2);
    \draw [green,  thick] (B1) to (B2);
    \draw [green,  thick] (T1) to (T2);
    \draw [blue,  thick] (Dir1) [bend left = 10] to (Dir2);
    \draw [green,  thick] (D1) [bend right = 15] to (Dir2);
    \draw [green,  thick] (DS1) to (Dir2);
    \draw [Fuchsia, ->] (Dir1) to (T1);
    \draw [Fuchsia, ->] (D1) to (Dir1);
    \draw [Fuchsia, ->] (F1) to (D1);
    \draw [Fuchsia, ->] (S1) to (F1);
    \draw [Fuchsia, ->] (S1) to (DS1);
    \draw [Fuchsia, ->] (DS1) to (Dir1);
    \draw [Fuchsia, ->] (B1) to (S1);
    \draw [Fuchsia, ->] (B2) to (E2);
    \draw [Fuchsia, ->] (E2) to (Dir2);
    \draw [Fuchsia, ->] (Dir2) to (AS2);
    \draw [Fuchsia, ->] (AS2) to (Sec2);
    \draw [Fuchsia, ->] (Sec2) to (T2);
    \draw (-3,-2.2) ellipse (1.4cm and 2.4cm);
    \draw (0,-2.2)[red] ellipse (1.2cm and 2.5cm);
    \end{tikzpicture}
    \caption{\small Unidirectional flow: If the blue arrows denote identified flows connecting important classes, then the green arrows are constrained by monotonicity to lie between them. \label{fig:WnD}}
    \end{minipage}

    \label{fig:first-combined}
\end{figure}
% \vspace{-0.5cm}
% The second example deals with information flow between domains with security lattice structures of quite different shape.
%

As long as information flows \textit{unidirectionally} from colleges to the University, \textit{monotone functions} from the security classes of the college lattice $\ld{C}$ to university security lattice $\rid{U}$ \textit{suffice} to ensure secure information flow.
A function $\rid{\alpha}: \ld{C} \rightarrow \rid{U}$ is called \textit{monotone} if whenever $\ld{sc_1} ~\ld{\sqsubseteq}~ \ld{sc_2}$ in $\ld{C}$ then $\rid{\alpha}(\ld{sc_1}) ~\rd{\sqsubseteq}~ \rid{\alpha}(\ld{sc_2})$ in $\rid{U}$.\footnote{Note that it's not necessary to make the function $\rid{\alpha}$ total or onto.}
Monotonicity also constrains possible flows between classes of the two domains, once certain important flows between certain classes have been identified (see Fig. \ref{fig:WnD}).
Moreover, since monotone functions are closed under composition, one can chain them to create secure \textit{unidirectional} information flow connections through a series of administrative domains.
Monotonicity is a basic principle adopted for information flow analyses, e.g. \cite{liu2017fabric}.

% \vspace{-0.4cm}
\begin{figure}[t]
    \begin{minipage}[t]{.45\textwidth}
    \centering
    \begin{tikzpicture}[framed,->,node distance=0.9cm,on grid]
    \title{W and D}
    \node(T2)   {\scriptsize$\top 2$};
    \node(T1)  [xshift=-3cm]  {\scriptsize$\top 1$};
    \node(Dir1) [below of = T1] {\scriptsize$CollegePrincipal$};
    \node(D1) [below left of = Dir1] {\scriptsize$Dean\ (F)$};
    \node(F1) [below of = D1] {\scriptsize$Faculty$};
    \node(DS1) [xshift=-2.3cm,yshift=-1.7cm] {\scriptsize$Dean\ (S)$};
    \node(S1) [below right of = F1] {\scriptsize$Student$};
    \node(B1)  [below of=S1]  {\scriptsize$\bot 1$};
    \node(Sec2)  [below of = T2] {\scriptsize$Chancellor$};
    \node(AS2)  [below of=Sec2] {\scriptsize$Vice\ Chancellor$};
    \node(Dir2)  [below of=AS2] {\scriptsize$Dean(Colleges)$};
    \node(E2)  [below of=Dir2] {\scriptsize$Univ.Fac.$};
    \node(B2)  [below of=E2]  {\scriptsize$\bot 2$};
    \draw [blue,  thick] (S1) [bend right = 5] to (E2);
    \draw [green,  thick] (F1) [bend right = 5] to (E2);
    \draw [green,  thick] (B1) to (B2);
    \draw [green,  thick, ->] (T1) to (T2);
    \draw [brown,  thick, ->] (T2) [bend right=10] to (T1);
    \draw [blue,  thick] (Dir1) [bend left = 10] to (Dir2);
    \draw [green,  thick] (D1) to (Dir2);
    \draw [green,  thick] (DS1) to (Dir2);
    \draw [Fuchsia, ->] (Dir1) to (T1);
    \draw [Fuchsia, ->] (D1) to (Dir1);
    \draw [Fuchsia, ->] (F1) to (D1);
    \draw [Fuchsia, ->] (S1) to (F1);
    \draw [Fuchsia, ->] (S1) to (DS1);
    \draw [Fuchsia, ->] (DS1) to (Dir1);
    \draw [Fuchsia, ->] (B1) to (S1);
    \draw [Fuchsia, ->] (B2) to (E2);
    \draw [Fuchsia, ->] (E2) to (Dir2);
    \draw [Fuchsia, ->] (Dir2) to (AS2);
    \draw [Fuchsia, ->] (AS2) to (Sec2);
    \draw [Fuchsia, ->] (Sec2) to (T2);
    \draw [red,  thick] (Dir2) [bend right = 10] to (F1);
    \draw [brown,  thick] (B2) [bend left=20] to (B1);
    \draw [brown,  thick] (E2) [bend right=10] to (F1);
    \draw [brown,  thick] (AS2)  to (T1);
    \draw [brown,  thick] (Sec2) to (T1);
    % \draw [green,  thick] (T2) [bend right=20] to (T1);
    \draw (-3,-2.2) ellipse (1.4cm and 2.4cm);
    \draw (0,-2.2)[red] ellipse (1.2cm and 2.5cm);
    \end{tikzpicture}
    \caption{\small %Security lattice on left side is for college and for university it's on the right side.
    The blue/green and brown/red arrows define monotone functions in each direction.
    However, the red arrow highlights a flow that is a security violation.
    \label{fig:Bi-WnD}}
    \end{minipage}
    \quad \quad
    \begin{minipage}[t]{.45\textwidth}
    \centering
    \begin{tikzpicture}[framed,->,node distance=0.9cm,on grid]
    \title{W and D}
    \node(T2)   {\scriptsize$\top 2$};
    \node(T1)  [xshift=-3cm]  {\scriptsize$\top 1$};
    \node(Dir1) [below of = T1] {\scriptsize$CollegePrincipal$};
    \node(D1) [below left of = Dir1] {\scriptsize$Dean\ (F)$};
    \node(F1) [below of = D1] {\scriptsize$Faculty$};
    \node(DS1) [xshift=-2.3cm,yshift=-1.7cm] {\scriptsize$Dean\ (S)$};
    \node(S1) [below right of = F1] {\scriptsize$Student$};
    \node(B1)  [below of=S1]  {\scriptsize$\bot 1$};
    \node(Sec2)  [below of = T2] {\scriptsize$Chancellor$};
    \node(AS2)  [below of=Sec2] {\scriptsize$Vice\ Chancellor$};
    \node(Dir2)  [below of=AS2] {\scriptsize$Dean(Colleges)$};
    \node(E2)  [below of=Dir2] {\scriptsize$Univ.Fac.$};
    \node(B2)  [below of=E2]  {\scriptsize$\bot 2$};
    \draw [green,  thick] (S1) [bend right = 10] to (E2);
    \draw [green,  thick] (F1) [bend right = 10] to (Dir2);
    \draw [green,  thick] (B1) to (B2);
    \draw [green,  thick, ->] (T1) [bend right=10] to (T2);
    \draw [brown,  thick, ->] (T2) [bend right=10] to (T1);
    \draw [green,  thick] (Dir1) [bend right= 3] to (Sec2);
    \draw [green,  thick] (D1) [bend right=15] to (AS2);
    \draw [green,  thick] (DS1) [bend right=10] to (AS2);
    \draw [Fuchsia, ->] (Dir1) to (T1);
    \draw [Fuchsia, ->] (D1) to (Dir1);
    \draw [Fuchsia, ->] (F1) to (D1);
    \draw [Fuchsia, ->] (S1) to (F1);
    \draw [Fuchsia, ->] (S1) to (DS1);
    \draw [Fuchsia, ->] (DS1) to (Dir1);
    \draw [Fuchsia, ->] (B1) to (S1);
    \draw [Fuchsia, ->] (B2) to (E2);
    \draw [Fuchsia, ->] (E2) to (Dir2);
    \draw [Fuchsia, ->] (Dir2) to (AS2);
    \draw [Fuchsia, ->] (AS2) to (Sec2);
    \draw [Fuchsia, ->] (Sec2) to (T2);
    \draw [brown,  thick] (B2) [bend left=10] to (S1);
    \draw [brown,  thick] (Sec2) to (T1);
    \draw [brown,  thick] (AS2) to (Dir1);
    \draw [brown,  thick] (E2) [bend left=10] to (F1);
    \draw [brown,  thick] (Dir2) [bend left=10] to (D1);
    % \draw [green,  thick] (T2) [bend right=20] to (T1);
    \draw (-3,-2.2) ellipse (1.4cm and 2.4cm);
    \draw (0,-2.2)[red] ellipse (1.2cm and 2.5cm);
    \end{tikzpicture}
    \caption{\small %Security lattice on left side is for college and for university it's on the right side.
    The arrows define a secure and precise connection.
    However, the security classification escalates quickly in a few round-trips, when information can flow in both directions.
    \label{fig:Bi-WnD-yCC}}
    \end{minipage}
\end{figure}

However, when there is ``blowback'' of information, mere monotonicity is \textit{inadequate} for ensuring SIF.
Consider the bidirectional flow situation in Fig. \ref{fig:Bi-WnD}, where data return to the original domain.
Monotonicity of both functions $\rid{\alpha}: \ld{C} \rightarrow \rid{U}$ and $\ld{\gamma}: \rid{U} \rightarrow \ld{C}$ does \textit{not} suffice for security because the composition $\ld{\gamma} \circ \rid{\alpha}$ may \textit{not} be non-decreasing.
In Fig. \ref{fig:Bi-WnD}, both $\rid{\alpha}$ and $\ld{\gamma}$ are monotone but their composition
can lead to information leaking from a higher class, e.g., \textit{College Principal}, to a lower class, e.g., \textit{Faculty} within $\ld{C}$ --- an outright violation of the college's security policy.
Similarly, composition $\rid{\alpha} \circ \ld{\gamma}$ may lead to violation of the University's security policy.
% \vspace{-0.2cm}
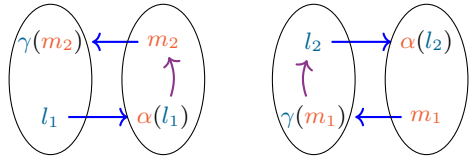
\begin{figure*}[htb]
%\begin{minipage}[b]{.5\textwidth}
\centering
\begin{tikzpicture}[->,node distance=1cm,on grid]
% Options for reducing size of text - \tiny \scriptsize \small \footnotesize
\title{Sec cond. 1}
\node(T2)   {$\rid{m_2}$};
\node(T1)  [xshift=-1.5cm]  {$\ld{\gamma}(\rid{m_2})$};
\node(B1)  [below of=T1]  {$\ld{l_1}$};
\node(B2)  [below of=T2]  {$\rid{\alpha}(\ld{l_1})$};

\draw [blue,  thick] (B1) to (B2);
\draw [Fuchsia, thick] (B2) to [bend right=20] (T2);
\draw [blue,  thick] (T2) to (T1);

\draw (-1.5,-0.5) ellipse (0.55cm and 1cm);
\draw (0,-0.5) ellipse (0.55cm and 1cm);

\node(T3)  [xshift=2cm]  {$\ld{l_2}$};
\node(T4)  [xshift=3.5cm]  {$\rid{\alpha}(\ld{l_2})$};
\node(B3)  [below of=T3]  {$\ld{\gamma}(\rid{m_1})$};
\node(B4)  [below of=T4]  {$\rid{m_1}$};
\draw [blue,  thick] (B4) to (B3);
\draw [Fuchsia, thick] (B3) to [bend left=20] (T3);
\draw [blue,  thick] (T3) to (T4);
\draw (2,-0.5) ellipse (0.55cm and 1cm);
\draw (3.5,-0.5) ellipse (0.55cm and 1cm);
\end{tikzpicture}
\caption{\small Secure flow conditions: (\textbf{sc1}) $\ld{l_1} \ld{\sqsubseteq} \ld{\gamma}(\rid{m_2})$ ~~(\textbf{sc2}) $\rid{m_1} \rd{\sqsubseteq} \rid{\alpha}(\ld{l_2})$.\label{fig:sec11}}
%\end{minipage}
\end{figure*}
% \vspace{-0.2cm}

\paragraph{Requirements.} \ \
We want to ensure that any ``round-trip'' flow of information, e.g., from a domain $\ld{L}$ to $\rid{M}$ and back to $\ld{L}$, is a permitted flow in the lattice $\ld{L}$, from where the data originated.
Thus we require the following (tersely stated) ``security conditions'' \textbf{SC1} and \textbf{SC2} on $\rid{\alpha}: \ld{L} \rightarrow \rid{M}$ and $\ld{\gamma}: \rid{M} \rightarrow \ld{L}$, which preclude any violation of the security policies of both the administrative domains
(see Fig. \ref{fig:sec11}):
\[
\textbf{SC1}~~ \lambda \ld{l}.\ld{l} ~\ld{\sqsubseteq}~
\ld{\gamma} \circ \rid{\alpha}
~~~~~~\hfill~~~~~~
\textbf{SC2} ~~  \lambda \rid{m}.\rid{m} ~ \rid{\sqsubseteq}~
\rid{\alpha} \circ \ld{\gamma}
\]
We also desire \textit{precision}, based on a principle of least privilege escalation ---
if data are exchanged between the two domains without any computation done on them, then the security level should not be needlessly raised.
Precision is important for meaningful and useful analyses.
\[
\begin{array}{c}
\textbf{PC1}~~\rid{\alpha}(\ld{l_1}) = \rid{\bigsqcup} ~ \{\rid{m_1} ~|~ \ld{\gamma}(\rid{m_1}) = \ld{l_1} \}, \; \; \forall \ld{l_1} \in \ld{\gamma}[\rid{M}]\\
\textbf{PC2} ~~\ld{\gamma}(\rid{m_1}) = \ld{\bigsqcup} ~ \{\ld{l_1} ~|~ \rid{\alpha}(\ld{l_1}) = \rid{m_1}\}, \; \; \forall \rid{m_1} \in \rid{\alpha}[\ld{L}]
\end{array}
\]
Further, if the data were to go back and forth between two domains more than once, the security classes to which data belong should not become increasingly restrictive after consecutive bidirectional data sharing (See Fig. \ref{fig:Bi-WnD-yCC}, which shows monotone functions that keep climbing up to the top).
This convergence requirement may be stated informally as conditions \textbf{CC1} and \textbf{CC2}, requiring  \textit{fixed points} for the compositions $\ld{\gamma} \circ \rid{\alpha}$ and $\rid{\alpha} \circ \ld{\gamma}$.
Since security lattices are finite, \textbf{CC1} and \textbf{CC2} necessarily hold -- such fixed points exist, though perhaps only at the topmost elements of the lattice.
We would therefore desire a stronger requirement, where fixed points are reached
as low in the orderings as possible. \\

\vspace{-0.37cm}
\begin{figure}
    \begin{minipage}[t]{.45\textwidth}
    \centering
    \begin{tikzpicture}[framed,->,node distance=1cm,on grid]
    \title{W and D}
    \node(T2)   {\scriptsize$\top 2$};
    \node(T1)  [xshift=-3cm]  {\scriptsize$\top 1$};
    \node(Dir1) [below of = T1] {\scriptsize$CollegePrincipal$};
    \node(D1) [below left of = Dir1] {\scriptsize$Dean\ (F)$};
    \node(F1) [below of = D1] {\scriptsize$Faculty$};
    \node(DS1) [xshift=-2.3cm,yshift=-1.7cm] {\scriptsize$Dean\ (S)$};
    \node(S1) [below right of = F1] {\scriptsize$Student$};
    \node(B1)  [below of=S1]  {\scriptsize$\bot 1$};
    \node(Sec2)  [below of = T2] {\scriptsize$Chancellor$};
    \node(AS2)  [below of=Sec2] {\scriptsize$Vice\ Chancellor$};
    \node(Dir2)  [below of=AS2] {\scriptsize$Dean(Colleges)$};
    \node(E2)  [below of=Dir2] {\scriptsize$Univ.Fac.$};
    \node(B2)  [below of=E2]  {\scriptsize$\bot 2$};
    \draw [green,  thick] (S1) [bend left = 5] to (Dir2);
    \draw [red,  thick] (Dir2) [bend left = 5] to (S1);
    \draw [red,  thick] (F1) to (Dir2);
    \draw [green,  thick] (B1) to (B2);
    \draw [green,  thick, ->] (T1) [bend right=10] to (T2);
    \draw [brown,  thick, ->] (T2) [bend right=10] to (T1);
    \draw [green,  thick] (Dir1) [bend left= 10] to (AS2);
    \draw [red,  thick] (D1) to (Dir2);
    \draw [red,  thick] (DS1) to (Dir2);
    \draw [Fuchsia, ->] (Dir1) to (T1);
    \draw [Fuchsia, ->] (D1) to (Dir1);
    \draw [Fuchsia, ->] (F1) to (D1);
    \draw [Fuchsia, ->] (S1) to (F1);
    \draw [Fuchsia, ->] (S1) to (DS1);
    \draw [Fuchsia, ->] (DS1) to (Dir1);
    \draw [Fuchsia, ->] (B1) to (S1);
    \draw [Fuchsia, ->] (B2) to (E2);
    \draw [Fuchsia, ->] (E2) to (Dir2);
    \draw [Fuchsia, ->] (Dir2) to (AS2);
    \draw [Fuchsia, ->] (AS2) to (Sec2);
    \draw [Fuchsia, ->] (Sec2) to (T2);
    \draw [brown,  thick] (B2) [bend left=20] to (B1);
    \draw [brown,  thick] (Sec2) to (T1);
    \draw [brown,  thick] (AS2) to (Dir1);
    \draw [brown,  thick] (E2) [bend left=10] to (S1);

    \draw (-3,-2.5) ellipse (1.4cm and 2.7cm);
    \draw (0,-2.5)[red] ellipse (1.2cm and 2.7cm);
    \end{tikzpicture}
    \caption{\small %Security lattice on left side is for college and for university it's on the right side.
    The arrows define a Galois Connection.
    However, the red arrows highlight flow security violations when information can flow in both directions.
    \label{fig:Bi-WnD-yNotGC}}
    \end{minipage}
    \quad \quad
    \begin{minipage}[t]{.45\textwidth}
    \centering
    \begin{tikzpicture}[framed,->,node distance=1cm,on grid]
    \title{W and D}
    \node(T2)   {{\color{blue}\scriptsize$\top 2$}};
    \node(T1)  [xshift=-3cm]  {{\color{blue} \scriptsize$\top 1$}};
    \node(Dir1) [below of = T1] {{\color{blue}\scriptsize$CollegePrincipal$}};
    \node(D1) [below left of = Dir1] {\scriptsize$Dean\ (F)$};
    \node(F1) [below of = D1] {\scriptsize$Faculty$};
    \node(DS1) [xshift=-2.3cm,yshift=-1.7cm] {\scriptsize$Dean\ (S)$};
    \node(S1) [below right of = F1] {{\color{blue}\scriptsize$Student$}};
    \node(B1)  [below of=S1]  {{\color{blue}\scriptsize$\bot 1$}};
    \node(Sec2)  [below of = T2] {\scriptsize$Chancellor$};
    \node[align=center] (AS2)  [below of=Sec2] {\scriptsize $Vice$ \\ \scriptsize $Chancellor$};
    \node(Dir2)  [below of=AS2] {{\color{blue}\scriptsize$Dean(Colleges)$}};
    \node(E2)  [below of=Dir2] {{\color{blue}\scriptsize$Univ.Fac.$}};
    \node(B2)  [below of=E2]  {{\color{blue}\scriptsize$\bot 2$}};
    \draw [black,  thick, <->] (F1) to (E2);
    \draw [OliveGreen, ->] (S1) to (E2);
    \draw [black,  thick,<->] (B1) to (B2);
    \draw [black,  thick, <->] (T1) to (T2);
    \draw [black,  thick, <->] (Dir1) [bend left = 10] to (Dir2);
    \draw [OliveGreen] (D1) to (Dir2);
    \draw [OliveGreen] (DS1) to (Dir2);
    \draw [OliveGreen] (AS2) to (T1);
    \draw [OliveGreen] (Sec2) to (T1);
    \draw [Fuchsia, ->] (Dir1) to (T1);
    \draw [Fuchsia, ->] (D1) to (Dir1);
    \draw [Fuchsia, ->] (F1) to (D1);
    \draw [Fuchsia, ->] (S1) to (F1);
    \draw [Fuchsia, ->] (S1) to (DS1);
    \draw [Fuchsia, ->] (DS1) to (Dir1);
    \draw [Fuchsia, ->] (B1) to (S1);
    \draw [Fuchsia, ->] (B2) to (E2);
    \draw [Fuchsia, ->] (E2) to (Dir2);
    \draw [Fuchsia, ->] (Dir2) to (AS2);
    \draw [Fuchsia, ->] (AS2) to (Sec2);
    \draw [Fuchsia, ->] (Sec2) to (T2);

    \draw (-3,-2.5) ellipse (1.4cm and 2.7cm);
    \draw (0,-2.5)[red] ellipse (1.2cm and 2.7cm);

    \draw [dotted] (0,-1.4) ellipse (0.8cm and 1.2cm);
    \draw [dotted] (-3,-1.7) ellipse (1.35cm and 0.3cm);
    \draw [dotted] (-3,-3.2) ellipse (0.5cm and 0.55cm);

    \end{tikzpicture}
    \caption{\small A useful increasing Lagois connection for sharing data. %between a college (left) to the university (on right side).
    Black arrows define permissible flows between buds. \label{fig:Bi-L-WnD}} %Negotiating a security MoU / increasing Lagois connection.
    \end{minipage}
\end{figure}

\paragraph{Galois connections aren't the answer.} \ \
Any discussion on a pair of partial orders linked by a pair of monotone functions suggests the notion of a Galois connection, an elegant and ubiquitous mathematical structure that finds use in computing, particularly in static analyses.
However, Galois connections are not the appropriate structure for bidirectional informational flow control.

Let $\ld{L}$ and $\rid{M}$ be two complete security class lattices, and $\rid{\alpha}: \ld{L} \rightarrow \rid{M}$ and $\ld{\gamma}: \rid{M} \rightarrow \ld{L}$ be two monotone functions such that $(\ld{L}, \rid{\alpha}, \ld{\gamma}, \rid{M})$ forms a Galois connection.
Recall that a Galois connection satisfies the condition
\[
    \textbf{GC1}~~
    \forall \ld{l_1} \in \ld{L}, \rid{m_1} \in \rid{M},\ \ \ \rid{\alpha}(\ld{l_1}) ~\rd{\sqsubseteq}~ \rid{m_1} ~\iff~
    \ld{l_1} ~\ld{\sqsubseteq}~ \ld{\gamma}(\rid{m_1})
\]
So in a Galois connection we have
$\rid{\alpha}(\ld{\gamma}(\rid{m_1})) ~ \rd{\sqsubseteq}~ \rid{m_1} \iff \ld{\gamma}(\rid{m_1}) ~\ld{\sqsubseteq}~ \ld{\gamma}(\rid{m_1})$.
Since $\ld{\gamma}(\rid{m_1}) ~\ld{\sqsubseteq}~ \ld{\gamma}(\rid{m_1})$ holds trivially, we get $\rid{\alpha}(\ld{\gamma}(\rid{m_1}))~ \rd{\sqsubseteq}~ \rid{m_1}$.
If $\rid{\alpha}(\ld{\gamma}(\rid{m_1})) ~\rd{\neq}~ \rid{m_1}$ then $ \rid{\alpha}(\ld{\gamma}(\rid{m_1})) ~\rd{\sqsubset}~ \rid{m_1}$ (strictly), which would violate secure flow requirement \textbf{SC2}.  Fig. \ref{fig:Bi-WnD-yNotGC} illustrates such a situation.

\paragraph{Why not Galois insertions?}
Now suppose $\ld{L}$ and $\rid{M}$ are two complete security class lattices, and $\rid{\alpha}: \ld{L} \rightarrow \rid{M}$ and $\ld{\gamma}: \rid{M} \rightarrow \ld{L}$ be two monotone functions such that $(\ld{L}, \rid{\alpha}, \ld{\gamma}, \rid{M})$ forms a \textit{Galois insertion}.
Then the flow of information permitted by $\rid{\alpha}$ and  $\ld{\gamma}$ is guaranteed to be secure.
However, Galois insertions mandate conditions on the definitions of functions $\rid{\alpha}$ and $\ld{\gamma}$ that are much too strong, i.e.,
\begin{itemize}
\item $\ld{\gamma}: \rid{M} \rightarrow \ld{L}$ is \textit{injective}, i.e.,
$\forall \rid{m_1,m_2} \in \rid{M} : \ld{\gamma}(\rid{m_1}) ~\ld{=}~ \ld{\gamma}(\rid{m_2}) \implies \rid{m_1} = \rid{m_2}$
\item $\rid{\alpha}: \ld{L} \rightarrow \rid{M}$ is \textit{surjective}, i.e.,
    $\forall \rid{m_1} \in \rid{M}, \exists \ld{l_1} \in \ld{L}:  \rid{\alpha}(\ld{l_1}) = \rid{m_1}$.
\end{itemize}
Typically data are shared only from a few security classes of any organisation.
Organisations rarely make public their entire security class structure and permitted flow policies.
Organisations also typically do not want any external influences on some subsets of its security classes.
Thus, if not all elements of $\rid{M}$ are transfer classes, it may be impossible to define a Galois insertion
$(\ld{L}, \rid{\alpha}, \ld{\gamma}, \rid{M})$ because we cannot force $\rid{\alpha}$ to be surjective.

\paragraph{Lagois Connections.} \ \
Further, the connection we seek to make between two domains should allow us to transpose them. % ``left-right transposable''.
Fortunately there is an elegant structure, i.e.,  \textit{Lagois Connections} \cite{MELTON1994lagoisconnections}, which exactly satisfies this as well as the requirements of security and bidirectional sharing (\textbf{SC1, SC2, PC1, PC2, CC1} and \textbf{CC2}).
They also conveniently generalise Galois Insertions.

\begin{definition}[Lagois Connection \cite{MELTON1994lagoisconnections}]
If $L = (\ld{L},\ld{\sqsubseteq})$ and $M = (\rid{M},\rd{\sqsubseteq})$ are two partially ordered sets, and $\rid{\alpha}: \ld{L} \rightarrow \rid{M}$ and $\ld{\gamma}: \rid{M} \rightarrow \ld{L}$ are order-preserving functions, then we call the quadruple $(\ld{L}, \rid{\alpha}, \ld{\gamma}, \rid{M})$ an {\em increasing} Lagois connection, if it satisfies the following properties:
\[
\begin{array}{llcll}
\textbf{LC1}~~ & \lambda \ld{l}.\ld{l} ~\ld{\sqsubseteq}~
\ld{\gamma} \circ \rid{\alpha}
& ~~~~~~~~~~~ &
\textbf{LC2}~~ & \lambda \rd{m}.\rd{m} ~\rd{\sqsubseteq}~
\rid{\alpha} \circ \ld{\gamma} \\
\textbf{LC3}~~ &  \rid{\alpha} \circ  \ld{\gamma}  \circ \rid{\alpha} = \rid{\alpha}
& ~~~~~~~~~~~ &
\textbf{LC4}~~ & \ld{\gamma}  \circ \rid{\alpha}  \circ \ld{\gamma} = \ld{\gamma}
\end{array}
\]
\end{definition}

\textbf{LC3} ensures that $\ld{\gamma}(\rid{\alpha}(\ld{c_1}))$ is the least upper bound of all security classes in $\ld{C}$ that are mapped to the same security class, say $\rid{u_1} = \rid{\alpha}(\ld{c_1})$ in $\rid{U}$.

The main result of this section is that  if the negotiated monotone functions $\rid{\alpha}$ and $\ld{\gamma}$ form a Lagois connection between the security lattices $\ld{L}$ and $\rid{M}$, then information flows permitted are secure and precise.

% \subsection{Lagois connection ensures SIF}
\begin{theorem}\label{thm:secureconnection}
Let $\ld{L}$ and $\rid{M}$ be two complete security class lattices, $\rid{\alpha}: \ld{L} \rightarrow \rid{M}$ and $\ld{\gamma}: \rid{M} \rightarrow \ld{L}$ be two monotone functions.
Then the flow of information permitted by $\rid{\alpha}$, $\ld{\gamma}$ satisfies conditions \textbf{SC1}, \textbf{SC2}, \textbf{PC1}, \textbf{PC2}, \textbf{CC1} and \textbf{CC2}
if $(\ld{L}, \rid{\alpha}, \ld{\gamma}, \rid{M})$ is an increasing Lagois connection.
\end{theorem}
\begin{proof}
Condition \textbf{SC1}  holds because
if $\rid{\alpha}(\ld{l_1}) ~\rd{\sqsubseteq}~ \rid{m_2}$,
by monotonicity of $\ld{\gamma}$,
$\ld{\gamma}(\rid{\alpha}(\ld{l_1})) ~\ld{\sqsubseteq}~ \ld{\gamma}(\rid{m_2})$.
But  by \textbf{LC1},
$\ld{l_1} ~\ld{\sqsubseteq}~ \ld{\gamma} (\rid{\alpha}(\ld{l_1}))$.
So $\ld{l_1} ~\ld{\sqsubseteq}~\ld{\gamma}(\rid{m_2})$.
(A symmetric argument holds for \textbf{SC2}.)
Conditions \textbf{PC1} and \textbf{PC2} are shown in Proposition 3.7 of  \cite{MELTON1994lagoisconnections}.
Conditions \textbf{CC1} and \textbf{CC2} hold since the compositions $\ld{\gamma} \circ \rid{\alpha}$ and $\rid{\alpha} \circ \ld{\gamma}$ are \textit{closure} operators, i.e., idempotent, extensive, order-preserving endo-functions on $\ld{L}$ and $\rid{M}$.
\end{proof}
In fact, Lagois connections ensure that information in a security class in the original domain remains accessible even after doing a round-trip from the other domain (Proposition 3.8 in \cite{MELTON1994lagoisconnections}):
\begin{align}
    \ld{\gamma}(\rid{\alpha}(\ld{l})) ~=~
    \ld{\sqcap} \{ \ld{l^*} \in \ld{\gamma}[\rid{M}] ~|~
    \ld{l} ~\ld{\sqsubseteq}~ \ld{l^*} \}, \label{TIGHT1} \\
     \rid{\alpha}( \ld{\gamma}(\rid{m})) ~=~
     \rid{\sqcap} \{ \rid{m^*} \in \rid{\alpha}[\ld{L}] ~|~
        \rid{m} ~\rd{\sqsubseteq}~ \rid{m^*} \}.
        \label{TIGHT2}
\end{align}
We list various propositions about Lagois connections, which illustrate some of their important properties.
In particular, the two functions $\ld{\gamma}$ and $\rid{\alpha}$ uniquely determine each other.  Moreover, there are largest members of their pre-images, which act as representatives for the equivalence classes of the equivalence relations $\rid{\thicksim_M}$ and $\ld{\thicksim_L}$ induced by these functions.
Further, since our security domains are complete lattices, these distinguished points are closed under meets.
\begin{proposition}[Proposition 3.7 in \cite{MELTON1994lagoisconnections}]
Let $(\ld{L}, \rid{\alpha}, \ld{\gamma}, \rid{M})$ be a Lagois connection and let $\rid{m} \in \rid{\alpha}[\ld{L}]$ and $\ld{l} \in \ld{\gamma}[\rid{M}]$. The $\rid{\alpha}^{-1}(\rid{m})$ has a largest member, which is $\ld{\gamma}(\rid{m})$, and $\ld{\gamma}^{-1}(\ld{l})$ has a largest member, which is $\rid{\alpha}(\ld{l})$.
\end{proposition}

This means that for all $\rid{m} \in \rid{\alpha}[\ld{L}]$ and $\ld{l} \in \ld{\gamma}[\rid{M}]$, $\ld{\gamma}(\rid{m})$ and $\rid{\alpha}(\ld{l})$ exist, and are defined by largest members in their pre-image. Also, the images $\ld{\gamma}[\rid{M}]$ and $\rid{\alpha}[\ld{L}]$ are isomorphic lattices.

%$L^* = \gamma[\alpha[L]] = \gamma[M]$ and $\thicksim_L$ and
$\rid{M^*} = \rid{\alpha}[\ld{\gamma}[\rid{M}]] = \rid{\alpha}[\ld{L}]$ defines a system of representatives for $\rid{\thicksim_M}$. Then, $\rid{\alpha}(\ld{\gamma}(\rd{m}))$ is the representative of the equivalence class $[\rd{m}]$ of $\rd{m}$ that lies in $\rid{M^*}$, called \textit{budpoint}, such that,
\begin{align}
    \textit{if}~ \rid{m}\in \rid{M} ~\textit{and}~ \rid{m^*} \in \rid{M^*} ~\textit{with}~ \rid{m} ~\rid{\thicksim_M}~ \rid{m^*}~\textit{then}~ \rid{m} ~\rd{\sqsubseteq}~ \rid{m^*}
\end{align}
Symmetrically, $\ld{L^*} = \ld{\gamma}[\rid{\alpha}[\ld{L}]] = \ld{\gamma}[\rid{M}]$ defines a system of representatives for $\ld{\thicksim_L}$.
\begin{proposition}[Proposition 3.9 in \cite{MELTON1994lagoisconnections}]
If $(\ld{L}, \rid{\alpha}, \ld{\gamma}, \rid{M})$ is a Lagois connection, then the functions $\rid{\alpha}$ and
$\ld{\gamma}$ uniquely determine each other; in fact
\begin{align}
    \ld{\gamma}(\rid{m}) ~=~ \ld{\bigsqcup}~ \rid{\alpha}^{-1}~[~\rid{\sqcap} \{~\rid{m^*} \in \rid{\alpha}[\ld{L}]~|~ \rid{m} ~\rd{\sqsubseteq}~ \rid{m^*}\} ~]\\
    \rid{\alpha}(\ld{l}) ~=~
    \rid{\bigsqcup}~ \ld{\gamma}^{-1} [~\ld{\sqcap} \{~ \ld{l^*} \in \ld{\gamma}~[~\rid{M}] ~|~
    \ld{l} ~\ld{\sqsubseteq}~ \ld{l^*} \} ~]
\end{align}
\end{proposition}
\begin{proposition}[Proposition 3.11 in \cite{MELTON1994lagoisconnections}]
\label{prop:meet-existence}
If $(\ld{L}, \rid{\alpha}, \ld{\gamma}, \rid{M})$ is a Lagois connection and $\ld{A} \subseteq \ld{\gamma}[\rid{M}]$, then
\begin{enumerate}
    \item the meet of $\ld{A}$ in $\ld{\gamma}[\rid{M}]$ exists if and only if the meet of $\ld{A}$ in $\ld{L}$ exists, and whenever either exists, they are equal.
    \item the join $\ld{\hat{a}}$ of $\ld{A}$ in $\ld{\gamma}[\rid{M}]$ exists if the join $\ld{\check{a}}$ of $\ld{A}$ in $\ld{L}$ exists, and in this case $\ld{\hat{a}} = \ld{\gamma}(\rid{\alpha}(\ld{\check{a}}))$
\end{enumerate}
\end{proposition}

\section{An Operational Model}\label{sec:model}

\subsection{Computational Model.}\label{sec:operations}
% \ref{TransactionalPL}.
Let us consider two different organisations $\ld{L}$ and $\rid{M}$ that want to share data with each other.
We start with the assumptions that the two domains comprise storage objects $\ld{Z}$ and $\rd{Z}$ respectively, which are manipulated using their own sets of \textit{atomic} transactional operations, ranged over by $\ld{t}$ and $\rd{t}$ respectively.
We further assume that these transactions within each domain are internally secure with respect to their flow models, and have no insecure or interfering interactions with the environment.
Thus, we are agnostic to the level of abstraction of the systems we aim to connect securely, and since our approach treats the application domains as ``black boxes'', it is readily adaptable to any level of discourse (language, system, OS, database) found in the security literature.

We extend these operations with a minimal set of operations to transfer data between the two domains.
To avoid any concurrency effects, interference or race conditions arising from inter-domain transfer, we augment the storage objects of both domains with a fresh set of \textit{export} and \textit{import} variables into/from which the data of the domain objects can be copied \textit{atomically}.
We designate these sets $\ld{X}, \rd{X}$ as the respective \textit{export} variables, and $\ld{Y}, \rd{Y}$ as the respective \textit{import} variables, with the corresponding variable instances written as $\ld{x_i}$, $\rd{x_i}$ and $\ld{y_i}$, $\rd{y_i}$.
These export and import variables form mutually disjoint sets, and are distinct from any extant domain objects manipulated by the applications within a domain.
These variables are used exclusively for transfer, and are manipulated atomically.
We let $\ld{w_i}$ range over all variables in $\ld{N} ~=~\ld{Z} \cup \ld{X} \cup \ld{Y}$ (respectively $\rd{w_i}$ over $\rd{N} ~=~ \rd{Z} \cup \rd{X} \cup \rd{Y}$).
Domain objects are copied \textit{to export} variables and \textit{from import} variables by special operations $\ld{rd}(\ld{z}, \ld{y})$ and $\ld{wr}(\ld{x}, \ld{z})$ (and $\rd{rd}(\rd{z}, \rd{y})$ and  $\rd{wr}(\rd{x}, \rd{z})$ in the other domain).
We assume \textit{atomic transfer} operations (\textit{trusted by both domains}) $T_{RL}, T_{LR}$ that copy data from the export variables of one domain to the import variables of the other domain as the only mechanism for inter-domain flow of data.
Let  ``phrase'' $p$ denote a command in either domain or a transfer operation, and let $s$ be any (empty or non-empty) sequence of phrases.

\[
\begin{array}{c}
% \text{(expressions)} & \ld{e} ::= \ld{w_i} ~|~ n &~~& \rd{e} ::= \rd{w_i} ~|~ n\\
\text{(command)} ~~~
    \ld{c} ::=  \ld{t} ~|~ \ld{rd}(\ld{z}, \ld{y}) ~|~  \ld{wr}(\ld{x}, \ld{z}) ~~~~~~
    \rd{c} ::=
    \rd{t} ~|~ \rd{rd}(\rd{z}, \rd{y}) ~|~  \rd{wr}(\rd{x}, \rd{z}) \\
    \text{(phrase)} ~~~ p ::= T_{RL}(\rd{x},\ld{y}) ~|~ T_{LR}(\ld{x},\rd{y}) ~|~
    \ld{c} ~|~ \rd{c}  ~~~\hfill~~~  \text{(seq)}~~~ s ::= \epsilon ~|~ s_1 ; p \\
\end{array}
\]

\begin{figure}
\[
\begin{array}{c}
\inferrule* [Left = \ld{T}]{\ld{\mu} \vdash \ld{t} \Rightarrow \ld{\nu}
}{\langle \ld{\mu},\rd{\mu} \rangle \vdash \ld{t} \Rightarrow \langle \ld{\nu}, \rd{\mu} \rangle}
~~~~\hfill~~~~
\inferrule* [Left = \rd{T}]{\rd{\mu} \vdash \rd{t} \Rightarrow \rd{\nu}
}{\langle \ld{\mu},\rd{\mu} \rangle \vdash \rd{t} \Rightarrow \langle \ld{\mu}, \rd{\nu} \rangle} \\[1ex]
\inferrule*[Left = \ld{Wr}]{% x \in dom(\ld{\mathcal{E}})\\z \in dom(\mu) \ld{\mathcal{E}} \subset \mu
}{\langle \ld{\mu},\rd{\mu} \rangle \vdash \ld{wr}(\ld{x},\ld{z}) \Rightarrow \langle \ld{\mu}[\ld{x} := \ld{\mu}(\ld{z})], \rd{\mu} \rangle}
\\
\inferrule*[Left = \rd{Wr}]{% x \in dom(\ld{\mathcal{E}})\\z \in dom(\mu) \ld{\mathcal{E}} \subset \mu
}{\langle \ld{\mu},\rd{\mu} \rangle \vdash \rd{wr}(\rd{x},\rd{z}) \Rightarrow \langle \ld{\mu}, \rd{\mu}[\rd{x} := \rd{\mu}(\rd{z})] \rangle}
\\[1ex]
\inferrule*[Left = \ld{Rd}]{%\ y \in dom(\ld{\mathcal{I}})\\ z \in dom(\mu) \\ \ld{\mathcal{I}} \subset \mu
}{\langle \ld{\mu},\rd{\mu} \rangle \vdash \ld{rd}(\ld{z},\ld{y}) \Rightarrow \langle \ld{\mu}[\ld{z} := \ld{\mu}(\ld{y})], \rd{\mu} \rangle} \\
\inferrule*[Left = \rd{Rd}]{%\ y \in dom(\ld{\mathcal{I}})\\ z \in dom(\mu) \\ \ld{\mathcal{I}} \subset \mu
}{\langle \ld{\mu},\rd{\mu} \rangle \vdash \rd{rd}(\rd{z},\rd{y}) \Rightarrow \langle \ld{\mu}, \rd{\mu}[\rd{z} := \rd{\mu}(\rd{y})] \rangle} \\[1ex]
\inferrule*[Left = Trl]{%\rd{x} \in dom(\rd{\mathcal{E}})\\y \in dom(\ld{\mathcal{I}})  \\ \ld{\mathcal{I}} \subset \mu \\ \rd{\mathcal{E}} \subset \rd{\mu}
}{\langle \ld{\mu},\rd{\mu} \rangle \vdash T_{RL}(\ld{y},\rd{x}) \Rightarrow
\langle \ld{\mu}[\ld{y} := \rd{\mu}(\rd{x})],\rd{\mu} \rangle} \\
\inferrule*[Left = Tlr]{%\rd{x} \in dom(\rd{\mathcal{E}})\\y \in dom(\ld{\mathcal{I}})  \\ \ld{\mathcal{I}} \subset \mu \\ \rd{\mathcal{E}} \subset \rd{\mu}
}{\langle \ld{\mu},\rd{\mu} \rangle \vdash T_{LR}(\rd{y},\ld{x}) \Rightarrow
\langle \ld{\mu},\rd{\mu}[\rd{y} := \ld{\mu}(\ld{x})] \rangle} \\[1ex]
\inferrule* [Left = Seq0]{}{
\langle \ld{\mu},\rd{\mu} \rangle \vdash \epsilon \Rightarrow^* \langle \ld{\mu},\rd{\mu} \rangle} \\
\inferrule* [Left = Seqs]{\langle \ld{\mu},\rd{\mu} \rangle \vdash s_1 \Rightarrow^* \langle \ld{\mu_1},\rd{\mu_1} \rangle, ~~~
\langle \ld{\mu_1},\rd{\mu_1} \rangle \vdash p \Rightarrow  \langle \ld{\mu_2},\rd{\mu_2} \rangle}{\langle \ld{\mu},\rd{\mu} \rangle \vdash s_1 ; p \Rightarrow^* \langle \ld{\mu_2},\rd{\mu_2} \rangle}
\end{array}
\]
\caption{Execution Rules}
    \label{fig:evalrules11}
\end{figure}

A \textit{store} (typically $\ld{\mu}, \ld{\nu}, \rd{\mu}, \rd{\nu}$) is a finite-domain function from variables to a set of values (not further specified).
We write, e.g., $\ld{\mu}(\ld{w})$ for the contents of the store $\ld{\mu}$ at variable $\ld{w}$, and $\ld{\mu}[\ld{w} := \rd{\mu}(\rd{w})]$ for the store
that is the same as $\ld{\mu}$ everywhere except at variable $\ld{w}$,
where it now takes value $\rd{\mu}(\rd{w})$.

The rules specifying execution of commands are given in Fig. \ref{fig:evalrules11}.
Assuming the specification of intradomain transactions of the form
$\ld{\mu} \vdash \ld{t} \implies \ld{\nu}$ and $\rd{\mu} \vdash \rd{t} \implies \rd{\nu}$, our rules allow us to specify judgments of the form
$\langle \ld{\mu}, \rd{\mu} \rangle \vdash p \implies \langle \ld{\nu}, \rd{\nu} \rangle$ for phrases, and (the reflexive-transitive closure) for sequences of phrases.
Note that phrase execution occurs \textit{atomically}, and the intra-domain transactions, as well as copying to and from the export/import variables affect the store in only one domain, whereas the \textit{atomic transfer} is only between export variables of one domain and the import variables of the other.

\subsection{Typing Rules}\label{sec:typing}
% TYPES
Let the two domains have the respective different IFMs:
\[ FM_L = \langle \ld{N}, \ld{P}, \ld{SC}, \ld{\sqcup}, \ld{\sqsubseteq} \rangle  ~~~\hfill~~~ FM_M = \langle \rd{N}, \rd{P}, \rd{SC}, \rid{\sqcup}, \rd{\sqsubseteq} \rangle, \]
such that the flow policies in both are defined over different sets of security classes $\ld{SC}$ and $\rd{SC}$.\footnote{Without loss of generality, we assume that $\ld{SC} \cap \rd{SC} = \emptyset$, since we can suitably rename security classes.}

The (security) types of the core language are as follows.

Metavariables $\ld{l}$ and $\rd{m}$ range over the sets of security classes, $\ld{SC}$ and $\rd{SC}$ respectively, which are partially ordered by $\ld{\sqsubseteq}$ and $\rd{\sqsubseteq}$.
A type assignment $\ld{\lambda}$ is a finite-domain function from variables $\ld{N}$ to $\ld{SC}$ (respectively, $\rd{\lambda}$ from $\rd{N}$ to $\rd{SC}$).
The important restriction we place on $\ld{\lambda}$
and $\rd{\lambda}$ is that they map export and import variables $\ld{X}, \rid{X}, \rd{Y}, \rid{Y}$ only to points in the security lattices $\ld{SC}$ and $\rd{SC}$ respectively which are in the domains of $\ld{\gamma}$ and $\rid{\alpha}$, i.e., these points participate in the Lagois connection.
Intuitively, a variable $w$ mapped to security class $\ld{l}$ can store information of security class $\ld{l}$ or lower.
The type system works with respect to given type assignment.  Given the security level, e.g., $\ld{l}$, the typing rules track for each command \textit{within that domain} whether all written-to variables in that domain are of security classes ``above'' $\ld{l}$, and additionally for transactions within a domain, they ensure  ``simple security'', i.e., that  all variables which may have been read belong to security classes ``below'' $\ld{l}$.
We assume for the transactions within a domain, e.g., $\ld{L}$, we have a type system that will give us judgments of the form $\ld{\lambda} \vdash \ld{c}: \ld{l}$.
The novel extension of our approach is to extend this framework to work over two connected domains, i.e., given implicit security levels of the contexts in the respective domains.
Cross-domain transfers will require pairing such judgments, and thus our type systems will have judgments of the form
\[ \langle \ld{\lambda}, \rd{\lambda} \rangle \vdash p: \langle
\ld{l}, \rd{m} \rangle \]

We introduce a syntax-directed set of typing rules for the core language, given in Fig. \ref{fig:syntax-type-rules11}.
In many of the rules, the type for one of the domains is not constrained by the rule, and so any suitable type may be chosen as determined by the context,  e.g., $\rd{m}$ in the rules \ld{\sc Tt},
\ld{\sc Trd}, \ld{\sc Twr} and $\ld{TT_{RL}}$,
and both $\ld{l}$ and $\rd{m}$ in {\sc Com0}.

\begin{figure}[t]
\[
\begin{array}{l}
\inferrule*[Left = \ld{Tt}]{}{
\langle \ld{\lambda}, \rd{\lambda} \rangle \vdash \ld{t}: \langle \ld{l}, \rd{m} \rangle ~\text{if for all $\ld{z}$ assigned in $\ld{t}$, }
\ld{l} ~\ld{\sqsubseteq}~ \ld{\lambda}(\ld{z})  \\
~~~~~~~~~~~~~~~~~~~~\hfill
\text{ \& for all } \ld{z_1} \text{ read in $\ld{t}$, }
\ld{\lambda}(\ld{z_1}) ~\ld{\sqsubseteq}~ \ld{l}}
\\
\inferrule*[Left = \rd{Tt}]{}{
\langle \ld{\lambda}, \rd{\lambda} \rangle \vdash \rd{t}: \langle \ld{l}, \rd{m} \rangle ~\text{if for all $\rd{z}$ assigned in $\rd{t}$, }
\rd{m} ~\rd{\sqsubseteq}~ \rd{\lambda}(\rd{z}) \\
~~~~~~~~~~~~~~~~~~~~\hfill
\text{ \& for all } \rd{z_1} \text{ read in $\rd{t}$, }
\rd{\lambda}(\rd{z_1}) ~\rd{\sqsubseteq}~ \rd{m}}
\\[1ex]
\inferrule*[Left = \ld{Trd}]{\ld{\lambda}(\ld{y}) ~\ld{\sqsubseteq}~ \ld{\lambda}(\ld{z})}{
\langle \ld{\lambda}, \rd{\lambda} \rangle \vdash \ld{rd}(\ld{z},\ld{y}):
\langle \ld{\lambda}(\ld{z}), \rd{m} \rangle
} \\
\inferrule*[Left = \rd{Trd}]{\rd{\lambda}(\rd{y}) ~\rd{\sqsubseteq}~ \rd{\lambda}(\rd{z})}{
\langle \ld{\lambda}, \rd{\lambda} \rangle \vdash \rd{rd}(\rd{z},\rd{y}):
\langle \ld{l}, \rd{\lambda}(\rd{z}) \rangle
} \\[1ex]
\inferrule*[Left = \ld{Twr}]{\ld{\lambda}(\ld{z}) ~\ld{\sqsubseteq}~ \ld{\lambda}(\ld{x})}{
\langle \ld{\lambda}, \rd{\lambda} \rangle \vdash \ld{wr}(\ld{x},\ld{z}):
\langle \ld{\lambda}(\ld{x}), \rd{m} \rangle} \\
\inferrule*[Left = \rd{Twr}]{\rd{\lambda}(\rd{z}) ~\rd{\sqsubseteq}~ \rd{\lambda}(\rd{x})}{
\langle \ld{\lambda}, \rd{\lambda} \rangle \vdash \rd{wr}(\rd{x},\rd{z}):
\langle \ld{l}, \rd{\lambda}(\rd{x}) \rangle} \\[1ex]
\inferrule*[Left = \ld{$TT_{RL}$}]{
\ld{\gamma}(\rd{\lambda}(\rd{x})) ~\ld{\sqsubseteq}~ \ld{\lambda}(\ld{y})
}{
\langle \ld{\lambda}, \rd{\lambda} \rangle \vdash T_{RL}(\ld{y},\rd{x}):
\langle \ld{\lambda}(\ld{y}), \rd{\lambda}(\rd{x}) \rangle} \\
\inferrule*[Left = \rd{$TT_{LR}$}]{
\rid{\alpha}(\ld{\lambda}(\ld{x})) ~\rd{\sqsubseteq}~ \rd{\lambda}(\rd{y})
}{
\langle \ld{\lambda}, \rd{\lambda} \rangle \vdash T_{LR}(\rd{y},\ld{x}):
\langle \ld{\lambda}(\ld{x}), \rd{\lambda}(\rd{y}) \rangle }\\[1ex]
\inferrule*[Left = Com0]{}{
\langle \ld{\lambda}, \rd{\lambda} \rangle \vdash \epsilon: \langle \ld{l}, \rd{m} \rangle} \\
\inferrule*[Left = ComP]{
\langle \ld{\lambda}, \rd{\lambda} \rangle \vdash p: \langle \ld{l_1}, \rd{m_1} \rangle
~~~~~\langle \ld{\lambda}, \rd{\lambda} \rangle \vdash s: \langle \ld{l}, \rd{m} \rangle
}{\langle \ld{\lambda}, \rd{\lambda} \rangle \vdash s; p:
\langle \ld{l_1} \ld{\sqcap} \ld{l} , \rd{m_1} \rid{\sqcap} \rd{m} \rangle}
\end{array}
    \]
    \caption{Typing rules}
    \label{fig:syntax-type-rules11}
\end{figure}

For transactions e.g., $\ld{t}$ entirely within domain $\ld{L}$, the typing rule \ld{\sc Tt} constrains the type in the left domain to be at a level $\ld{l}$ that dominates all variables read in
$\ld{t}$, and which is dominated by all variables written to in $\ld{t}$, but places no constraints on the type $\rd{m}$ in the other domain $\rid{M}$.
In the rule \ld{\sc Trd}, since a value in import variable $\ld{y}$ is copied to the variable $\ld{z}$, we have $\ld{\lambda}(\ld{y}) ~\ld{\sqsubseteq}~ \ld{\lambda}(\ld{z})$, and the type in the domain $\ld{L}$ is $\ld{\lambda}(\ld{z})$ with no constraint on the type $\rd{m}$ in the other domain.
Conversely, in the rule \ld{\sc Twr}, since a value in variable $\ld{z}$ is copied to the export variable $\ld{x}$, we have $\ld{\lambda}(\ld{z}) ~\ld{\sqsubseteq}~ \ld{\lambda}(\ld{x})$, and the type in the domain $\ld{L}$ is $\ld{\lambda}(\ld{x})$ with no constraint on the type $\rd{m}$ in the other domain.
In the rule $\ld{TT_{RL}}$, since the contents of
a variable $\rd{x}$ in domain $\rid{M}$ are copied into a variable $\ld{y}$ in domain $\ld{L}$, we require $\ld{\gamma}(\rd{\lambda}(\rd{x})) ~\ld{\sqsubseteq}~ \ld{\lambda}(\ld{y})$, and constrain the type in domain $\ld{L}$ to $\ld{\lambda}(\ld{y})$.
The constraint in the other domain is unimportant (but for the sake of convenience, we peg it at $\rd{\lambda}(\rd{x})$).
Finally, for the types of sequences of phrases, we take the meets of the collected types in each domain respectively, so that we can guarantee that no variable of type lower than these meets has been written into during the sequence.
Note that Proposition \ref{prop:meet-existence} ensures that these types have the desired properties for participating in the Lagois connection.

\subsection{Soundness} \label{sec:soundness}
We now establish soundness of our scheme by showing a non-interference theorem with respect to operational semantics and the type systems built on the security lattices.
This theorem may be viewed as a conservative adaptation (to a minimal secure data transfer framework in a Lagois-connected pair of domains) of the main result of Volpano \textit{et al} \cite{DBLP:journals/jcs/VolpanoIS96}.

We assume that underlying base transactional languages in each of the domains have the following simple property (stated for $\ld{L}$, but an analogous property is assumed for $\rid{M}$):
Within each transaction $\ld{t}$, for each assignment of an expression $\ld{e}$ to any variable $\ld{z}$,  the following holds:
If $\ld{\mu}$, $\ld{\nu}$ are two stores such that
for all $\ld{w} \in vars(\ld{e})$, we have $\ld{\mu}(\ld{w}) = \ld{\nu}(\ld{w})$, then after executing the assignment, we will get $\ld{\mu}(\ld{z}) = \ld{\nu}(\ld{z})$.
That is, if two stores are equal for all variables appearing in the expression $\ld{e}$, then the value assigned to the variable $\ld{z}$ will be the same.
This assumption plays the r\^{o}le of ``Simple Security'' of expressions in \cite{DBLP:journals/jcs/VolpanoIS96} in the proof of the main theorem.
The type system plays the r\^{o}le of ``Confinement''.
We start with two obvious lemmas about the operational semantics, namely preservation of domains, and a ``frame'' lemma:
\begin{lemma}[Domain preservation]\label{lemma:equaldomainoneval}
If $\langle \ld{\mu}, \rd{\mu} \rangle \vdash s \Rightarrow^* \langle \ld{\mu_1}, \rd{\mu_1} \rangle$, then $dom(\ld{\mu}) = dom(\ld{\mu_1})$, and $dom(\rd{\mu}) = dom(\rd{\mu_1})$.
\end{lemma}
\begin{proof}
By induction on the length of the derivation of $\langle \mu, \rd{\mu} \rangle \vdash s \Rightarrow^* \langle \mu_1, \rd{\mu_1} \rangle$.
\end{proof}
\begin{lemma}[Frame]\label{lemma:notassigned}
If $\langle \mu, \rd{\mu} \rangle \vdash s \Rightarrow^* \langle \mu_1, \rd{\mu_1} \rangle, w \in dom(\mu) \cup dom(\rd{\mu})$, and $w$ is \textit{not} assigned to in $s$, then $\ld{\mu}(w) = \ld{\mu_1}(w)$ and
$\rd{\mu}(w) = \rd{\mu_1}(w)$.
\end{lemma}
\begin{proof}
By induction on the length of the derivation of $\langle \ld{\mu}, \rd{\mu} \rangle \vdash s \Rightarrow^* \langle \ld{\mu_1}, \rd{\mu_1} \rangle$.
\end{proof}

The main result of the paper assumes an ``adversary'' that operates at a security level $\ld{l}$ in domain $\ld{L}$ and at security level $\rd{m}$ in domain $\rid{M}$.
Note however, that these two levels are interconnected by the monotone functions $\rid{\alpha}: \ld{L} \rightarrow \rid{M}$ and
$\ld{\gamma}: \rid{M} \rightarrow \ld{L}$, since these levels are connected by the ability of information at one level in one domain to flow to the other level in the other domain.
\begin{theorem}[Type Soundness]\label{thm:soundness}
Suppose $\ld{l}, \rd{m}$ are the ``adversarial'' type levels in the respective domains, which satisfy the condition $\ld{l} = \ld{\gamma}(\rd{m})$ and
$\rd{m} = \rid{\alpha}(\ld{l})$.
Let
\begin{enumerate}[label=(\alph*)]
    \item \label{assume:1}$\langle \ld{\lambda}, \rd{\lambda} \rangle \vdash s:
    \langle \ld{l_0}, \rd{m_0}\rangle$;
    \item \label{assume:2}$\langle \ld{\mu}, \rd{\mu} \rangle \vdash s \Rightarrow^* \langle \ld{\mu_f}, \rd{\mu_f} \rangle$;
    \item \label{assume:3}$\langle \ld{\nu}, \rd{\nu} \rangle \vdash s \Rightarrow^* \langle \ld{\nu_f}, \rd{\nu_f} \rangle$;
    \item \label{assume:4}$dom(\ld{\mu}) = dom(\ld{\nu}) = dom(\ld{\lambda})$ and
    $dom(\rd{\mu}) = dom(\rd{\nu}) = dom(\rd{\lambda})$;
    \item \label{assume:5} $\ld{\mu}(\ld{w}) = \ld{\nu}(\ld{w})$ for all
    $\ld{w}$ such that $\ld{\lambda}(\ld{w}) ~\ld{\sqsubseteq}~ \ld{l}$, and
    $\rd{\mu}(\rd{w}) = \rd{\nu}(\rd{w})$ for all
    $\rd{w}$ such that $\rd{\lambda}(\rd{w}) ~\rd{\sqsubseteq}~ \rd{m}$.

\end{enumerate}
Then
$\ld{\mu_f}(\ld{w}) = \ld{\nu_f}(\ld{w})$ for all
    $\ld{w}$ such that $\ld{\lambda}(\ld{w}) ~\ld{\sqsubseteq}~ \ld{l}$, and
    $\rd{\mu_f}(\rd{w}) = \rd{\nu_f}(\rd{w})$ for all
    $\rd{w}$ such that $\rd{\lambda}(\rd{w}) ~\rd{\sqsubseteq}~ \rd{m}$.
    \begin{comment}
Then,  $\langle \nu_f, \rd{\nu_f}] (w'') = [\mu_f, \rd{\mu_f}] (w'')$ for all w'' such that $[\lambda, \rd{\lambda}] (w'') \sqsubseteq (\tau \sqcup\  \gamma(\rd{\tau}))$ or $[\lambda, \rd{\lambda}] (w'') \sqsubseteq (\rd{\tau} \sqcup \alpha(\tau))$, where $w'' \in x,y,z,\rd{x},\rd{y},\rd{z}$
\end{comment}
\end{theorem}

\begin{proof}
By induction on the length of sequence $s$.
The base case is vacuously true.
We now consider a sequence $s_1; p$.
$\langle \ld{\mu},\rd{\mu} \rangle \vdash s_1 \Rightarrow^* \langle \ld{\mu_1},\rd{\mu_1} \rangle$ and
$\langle \ld{\mu_1},\rd{\mu_1} \rangle \vdash p \Rightarrow  \langle \ld{\mu_f},\rd{\mu_f} \rangle$
and
$\langle \ld{\nu},\rd{\nu} \rangle \vdash s_1 \Rightarrow^* \langle \ld{\nu_1},\rd{\nu_1} \rangle$ and
$\langle \ld{\nu_1},\rd{\nu_1} \rangle \vdash p \Rightarrow  \langle \ld{\nu_f},\rd{\nu_f} \rangle$
By induction hypothesis applied to $s_1$, we have
$\ld{\mu_1}(\ld{w}) = \ld{\nu_1}(\ld{w})$ for all
$\ld{w}$ such that $\ld{\lambda(\ld{w})} ~\ld{\sqsubseteq}~\ld{l}$,
and
$\rd{\mu_1}(\rd{w}) = \rd{\nu_1}(\rd{w})$ for all
$\rd{w}$ such that $\rd{\lambda}(\rd{w}) ~\rd{\sqsubseteq}~ \rd{m}$.

Let  $\langle \ld{\lambda}, \rd{\lambda} \rangle \vdash s_1: \langle \ld{l_s}, \rd{m_s} \rangle$, and $\langle \ld{\lambda}, \rd{\lambda} \rangle \vdash p: \langle \ld{l_p}, \rd{m_p} \rangle$.
We examine four cases for $p$ (the remaining cases are symmetrical). \\
\textbf{Case} $p$ is $\ld{t}$: \ \ \
Consider any $\ld{w}$ such that
$\ld{\lambda}(\ld{w}) ~\ld{\sqsubseteq}~\ld{l}$.
If $\ld{w} \in \ld{X} \cup \ld{Y}$ (i.e., it doesn't appear in $\ld{t}$), or if $\ld{w} \in \ld{Z}$ but is not assigned to in $\ld{t}$, then
by Lemma \ref{lemma:notassigned} and the induction hypothesis, $\ld{\mu_f}(\ld{w}) =
\ld{\mu_1}(\ld{w}) = \ld{\nu_1}(\ld{w}) =
\ld{\nu_f}(\ld{w})$. \\
Now suppose $\ld{z}$ is assigned to in $\ld{t}$.
From the condition $\langle \ld{\lambda}, \rd{\lambda} \rangle \vdash p: \langle \ld{l_p}, \rd{m_p} \rangle$, we know that
for all $\ld{z_1}$ assigned in $\ld{t}$,  $\ld{l_p} ~\ld{\sqsubseteq}~ \ld{\lambda}(\ld{z_1})$ and
for all $\ld{z_1}$ read in $\ld{t}$, $\ld{\lambda}(\ld{z_1}) ~\ld{\sqsubseteq}~ \ld{l_p}$.
Now if $\ld{l} \ld{~\sqsubseteq}~ \ld{l_p}$, then since in $\ld{t}$ no variables $\ld{z_2}$ such that
$\ld{\lambda}(\ld{z_2}) ~\ld{\sqsubseteq}~\ld{l}$ are assigned to.
Therefore  by Lemma \ref{lemma:notassigned}, $\ld{\mu_f}(\ld{w}) =  \ld{\mu_1}(\ld{w}) = \ld{\nu_1}(\ld{w}) = \ld{\nu_f}(\ld{w})$, for all $\ld{w}$ such that $\ld{\lambda}(\ld{w}) ~\ld{\sqsubseteq}~ \ld{l}$. \\
If $\ld{l_p} \ld{~\sqsubseteq}~ \ld{l}$, then for all $\ld{z_1}$ read in $\ld{t}$, $\ld{\lambda}(\ld{z_1}) ~\ld{\sqsubseteq}~ \ld{l_p}$.
Therefore, by assumption on transaction $\ld{t}$, if any variable $\ld{z}$  is assigned an expression $\ld{e}$,
since $\ld{\mu_1}$, $\ld{\nu_1}$ are two stores such that
for all $\ld{z_1} \in \ld{Z_e} = vars(\ld{e})$,  $\ld{\mu_1}(\ld{z_1}) = \ld{\nu_1}(\ld{z_1})$, the value of $\ld{e}$ will be the same.
By this simple security argument, after the transaction $\ld{t}$, we have $\ld{\mu_f}(\ld{z}) = \ld{\nu_f}(\ld{z})$.
Since the transaction happened entirely and atomically in domain $\ld{L}$, we do not have to worry ourselves with changes in the other domain
$\rid{M}$, and do not need to concern ourselves with
the adversarial level $\rd{m}$.\\[1ex]
\textbf{Case} $p$ is $\ld{rd}(\ld{z},\ld{y})$: \ \ \
Thus $\langle \ld{\lambda}, \rd{\lambda} \rangle \vdash \ld{rd}(\ld{z},\ld{y}):
\langle \ld{\lambda}(\ld{z}), \rd{m} \rangle$,
which means $\ld{\lambda}(\ld{y}) ~\ld{\sqsubseteq}~ \ld{\lambda}(\ld{z})$.
If $\ld{l} ~\ld{\sqsubseteq}~ \ld{\lambda}(\ld{z})$,
there is nothing to prove (Lemma \ref{lemma:notassigned}, again).
If $\ld{\lambda}(\ld{z}) ~\ld{\sqsubseteq}~ \ld{l}$, then since by I.H., $\ld{\mu_1}(\ld{y}) = \ld{\nu_1}(\ld{y})$, we have $\ld{\mu_f}(\ld{z}) =
\ld{\mu_1}[\ld{z} := \ld{\mu_1}(\ld{y})](\ld{z}) =
\ld{\nu_1}[\ld{z} := \ld{\nu_1}(\ld{y})](\ld{z}) =
\ld{\nu_f}(\ld{z})$. \\[1ex]
\textbf{Case} $p$ is $\ld{wr}(\ld{x},\ld{z})$: \ \ \
Thus $\langle \ld{\lambda}, \rd{\lambda} \rangle \vdash \ld{wr}(\ld{x},\ld{z}):
\langle \ld{\lambda}(\ld{x}), \rd{m} \rangle$,
which means $\ld{\lambda}(\ld{z}) ~\ld{\sqsubseteq}~ \ld{\lambda}(\ld{x})$.
If $\ld{l} ~\ld{\sqsubseteq}~ \ld{\lambda}(\ld{x})$,
there is nothing to prove (Lemma \ref{lemma:notassigned}, again).
If $\ld{\lambda}(\ld{x}) ~\ld{\sqsubseteq}~ \ld{l}$, then since by I.H., $\ld{\mu_1}(\ld{z}) = \ld{\nu_1}(\ld{z})$, we have $\ld{\mu_f}(\ld{x}) =
\ld{\mu_1}[\ld{x} := \ld{\mu_1}(\ld{z})](\ld{x}) =
\ld{\nu_1}[\ld{x} := \ld{\nu_1}(\ld{z})](\ld{x}) =
\ld{\nu_f}(\ld{x})$. \\[1ex]
\textbf{Case} $p$ is $T_{RL}(\ld{y},\rd{x})$: \ \ \
So $\langle \ld{\lambda}, \rd{\lambda} \rangle \vdash T_{RL}(\ld{y},\rd{x}):
\langle \ld{\lambda}(\ld{y}), \rd{\lambda}(\rd{x}) \rangle$, and
$\ld{\gamma}(\rd{\lambda}(\rd{x})) ~\ld{\sqsubseteq}~ \ld{\lambda}(\ld{y})$.
If $\ld{l} ~\ld{\sqsubseteq}~ \ld{\lambda}(\ld{y})$,
there is nothing to prove (Lemma \ref{lemma:notassigned}, again).
If $\ld{\lambda}(\ld{y}) ~\ld{\sqsubseteq}~\ld{l}$,
then by transitivity, $\ld{\gamma}(\rd{\lambda}(\rd{x})) ~\ld{\sqsubseteq}~ \ld{l}$.
By monotonicity of $\rid{\alpha}$:
$\rid{\alpha}(\ld{\gamma}(\rd{\lambda}(\rd{x}))) ~\rd{\sqsubseteq}~ \rid{\alpha}(\ld{l}) ~=~ \rd{m}$
(By our assumption on $\ld{l}$ and $\rd{m}$).
But by \textbf{LC2}, $\rd{\lambda}(\rd{x}) ~\rd{\sqsubseteq}~ \rid{\alpha}(\ld{\gamma}(\rd{\lambda}(\rd{x})))$.
So by transitivity, $\rd{\lambda}(\rd{x}) ~\rd{\sqsubseteq}~ \rd{m}$.
Now, by I.H., since $\rd{\mu_1}(\rd{x}) = \rd{\nu_1}(\rd{x})$, we have
$\ld{\mu_f}(\ld{y}) =
\ld{\mu_1}[\ld{y} := \rd{\mu_1}(\rd{x})](\ld{y}) =
\ld{\nu_1}[\ld{y} := \rd{\nu_1}(\rd{x})](\ld{y}) =
\ld{\nu_f}(\ld{y})$.
\end{proof}

\section{Related Work}\label{sec:related}
The notion of Lagois connections \cite{MELTON1994lagoisconnections} has surprisingly not been employed much in computer science.
The only cited use of this idea seems to be the work of
Huth \cite{huth1993equivalence} in establishing the correctness of programming language implementations.
To our knowledge, our work is the only one to propose their use in secure information flow control.

Abstract Interpretation and type systems \cite{cousot1997types-as-ai} have been used in secure flow analyses, e.g.,  \cite{cortesi2015datacentricsemantics, cortesi2018} and  \cite{zanotti2002sectypingsbyai}, where security types are defined using Galois connections employing, for instance, a standard collecting semantics.
Their use of two domains, concrete and abstract, with a Galois connection between them, for performing static analyses \textit{within a single domain} should not be confused with our idea of secure connections between independently-defined security lattices of two organisations.

There has been substantial work on SIF in a distributed setting at the systems level.
DStar \cite{zeldovich2008-nsdi} for example, uses  sets of opaque identifiers to define security classes.
The DStar framework extends a \textit{particular} DIFC model \cite{Krohn2007-aa,zeldovich2006-osdi} for operating systems to a distributed network.
The only partial order that is considered in DStar's security lattice is subset inclusion.
So it is not clear if DStar can work on general IFC mechanisms such as FlowCaml \cite{Pottier2003-FlowCaml}, which can use any partial ordering.
Nor can it express the labels of  JiF \cite{myers1999jflow} or Fabric \cite{liu2017fabric} completely.
DStar allows bidirectional communication between processes $R$ and $S$ only if $L_R \sqsubseteq_{O_R} L_S$ and $L_S \sqsubseteq_{O_S} L_R$, i.e., if there is an order-isomorphism between the labels.
Our motivating examples indicate such a requirement is far too restrictive for most practical arrangements for data sharing between organisations.

Fabric \cite{liu2009fabric,liu2017fabric} adds \textit{trust relationships} directly derived from a principal hierarchy to support federated systems with mutually distrustful nodes and allows dynamic delegation of authority.

Most of the previous DIFC mechanisms \cite{myers1999jflow, zeldovich2006-osdi, Krohn2007-aa, efstathopoulos2005asbestos, roy2009laminar, cheng2012aeolus} including Fabric are susceptible to the vulnerabilities illustrated in our motivating examples,  which we will mention in the concluding discussion.

\section{Conclusions and Future Work}\label{sec:conclusion}

Our work is  similar in spirit to Denning's motivation for proposing lattices, namely to identify a simple and mathematically elegant structure in which to frame the construction of scalable secure information flow in a modular manner that preserved the autonomy of the individual organisations.
From the basic requirements, we identified the elegant theory of Lagois connections as an appropriate structure.
Lagois connections provide us a way to connect the security lattices of two (secure) systems in a manner that does not expose their entire internal structure and allows us to reason only in terms of the interfaced security classes.
We believe that this framework is also applicable in more intricate information flow control formulations such as decentralised IFC \cite{myers-phd-tr-award} and models with declassification, as well as formulations with data-dependent security classes \cite{Lourenco2015-ug}.
We intend to explore these aspects in the future.

In this paper, we also proposed a minimal operational model for the transfer of data between the two domains.
This formulation is spare enough to be adaptable at various levels of abstraction (programming language, systems, databases), and is intended to illustrate that the Lagois connection framework can \textit{conserve} security, using non-interference as the semantic notion of soundness.
The choice of non-interference and the use of a type system in the manner of Volpano \textit{et al.} \cite{DBLP:journals/jcs/VolpanoIS96} was to illustrate in familiar terms how those techniques (removed from a particular language formulation) could be readily adapted to work in the context of secure connections between lattices.
In this exercise, we made suitable assumptions of atomicity and the use of fresh variables for communication, so as to avoid usual sources of interference.
We believe that the Lagois connection framework for secure flows between systems is readily adaptable for notions of  semantic correctness other than non-interference, though that is an exercise for the future.

In the future we intend to explore how the theory of Lagois connections constitutes a robust framework that can support the discovery, decomposition, update and maintenance of secure MoUs for exchanging information.
In this paper, we concerned ourselves only with two domains and bidirectional information exchange.
Compositionality of Lagois connections allows these results to extend to chaining connections across several domains.
In the future, we also intend to explore how one may secure more complicated information exchange arrangements than merely chains of bidirectional flow.

%\subsection{Future Work}
We close this discussion with a reminder of why it is important to have a framework in which secure flows should be treated in a modular and autonomous manner.
Consider Myer's DIFC model described in \cite{myers-phd-tr-award}, where a principal can delegate to others the capacity to act on its behalf.
We believe that this notion does not scale well to
large, networked systems since a principal may repose different levels of trust in the various hosts in the network.
For this reason, we believe that frameworks such as Fabric \cite{liu2009fabric, liu2017fabric} may provide more power than  mandated by a principle of least privilege.
In general, since a principal rarely vests unqualified trust in another in all contexts and situations, one should confine the influence of the principals possessing delegated authority to only specific domains.
A mathematical framework that can deal with localising trust and delegation of authority in different domains and controlling the manner in which information flow can be secured deserves a deeper study.
We believe that algebraic theories such as Lagois connections can provide the necessary structure for articulating these concepts.

\paragraph{Acknowledgments.} The second author thanks Deepak Garg for insightful discussions on secure information flow.  Part of the title is stolen from E.M. Forster.

\end{document}